\documentclass{article}
\usepackage{graphicx,epsfig,color}
\usepackage{graphics,latexsym}
\usepackage{enumerate}

\newcommand{\squeeze}{
\textwidth 6in \textheight 8.8in \oddsidemargin 0.2in \topmargin
-0.4in }

\squeeze

\newcommand{\ceiling}[1]{\left\lceil#1\right\rceil}
\newcommand{\scrod}{\quad\nopagebreak}

\newtheorem{lemma}{Lemma}

\newtheorem{theorem}[lemma]{Theorem}

\newtheorem{definition}{Definition}

\newenvironment{proof}
{\bigskip\noindent\textbf{Proof~}} {\marginpar{$\Box$}\bigskip}

\newcommand{\Si}{{\rm S}}

\newcommand{\RP}{{\rm RP}}

\newcommand{\subE}{{\rm subE}}

\newcommand{\NP}{{\rm NP}}

\newcommand{\Pnew}{{\rm P}}

\begin{document}



\date{}

\title{Multivariate Polynomial Integration and Derivative Are
 Polynomial Time Inapproximable Unless P=NP\thanks{This research is
supported in part by National Science Foundation Early Career Award
0845376.}}

\author{
Bin Fu
 \\ \\
Department of Computer Science\\
 University of Texas-Pan American\\
Edinburg, TX  78539, USA\\
binfu@cs.panam.edu\\\\
} \maketitle

\begin{abstract} We investigate the complexity of integration and
derivative for multivariate polynomials in the standard computation
model. The integration is in the unit cube $[0,1]^d$ for a
multivariate polynomial, which  has format $f(x_1,\cdots,
x_d)=p_1(x_1,\cdots, x_d)p_2(x_1,\cdots, x_d)\cdots p_k(x_1,\cdots,
x_d)$, where each $p_i(x_1,\cdots, x_d)=\sum_{j=1}^d q_j(x_j)$ with
all single variable polynomials $q_j(x_j)$ of  degree at most two
and constant coefficients. We show that there is no any factor
polynomial time approximation for the integration
$\int_{[0,1]^d}f(x_1,\cdots,x_d)d_{x_1}\cdots d_{x_d}$ unless
$\Pnew=\NP$.  For the complexity of multivariate derivative, we
consider the functions with the format $f(x_1,\cdots,
x_d)=p_1(x_1,\cdots, x_d)p_2(x_1,\cdots, x_d)\cdots p_k(x_1,\cdots,
x_d),$ where each $p_i(x_1,\cdots, x_d)$ is of degree at most $2$
and $0,1$ coefficients. We also show that unless $\Pnew=\NP$, there
is no any factor polynomial time approximation to its derivative
${\partial f^{(d)}(x_1,\cdots, x_d)\over
\partial x_1\cdots
\partial x_d}$ at the origin point $(x_1,\cdots, x_d)=(0,\cdots,0)$. Our $\#P$-hard result for derivative  shows that the
 derivative is not be easier than the integration in high
 dimension.  We also give some
tractable cases of high dimension integration and derivative.
\end{abstract}




\section{Introduction}

Integration and derivative are basic operations in the classical
mathematics. Integrations with a large number of variables have been
found applications in many areas such as finance, nuclear physics,
and quantum system, etc. The complexity for approximating
multivariate integration has been studied by measuring the number of
function evaluations. For example, Sloan and Wozniakowski proved an
exponential lower bounds $2^s$ of function evaluations in order to
obtain an approximation with error less than the integration itself,
which has $s$ variables~\cite{SloanWozniakowski97}. The integration
$\int_{[0,1]^s} f(x_1,\cdots,x_s) d_{x_1}\cdots d_{x_s}$ is over the
cubic $[0,1]^s$ for some function $f(x_1,\cdots, x_s)$.
In the quasi-Monte Carlo method for computing
$\int_{[0,1]^d}f(x)d_x$, it is approximated by ${1\over
n}\sum_{i=1}^n f(x_i)$. This approximation has an error
$\Theta({(\ln n)^{d-1}\over n})$, which grows exponentially on the
dimension number $d$ (see e.x.,~\cite{Niederreiter92,Niederreiter78}
and the reference papers there).

An integration may be computed by the structure of the function
without involving function evaluation. For example,
$\int_{[0,1]^2}x^2y^3d_xd_y=(\int_{[0,1]}x^2d_x)\cdot
(\int_{[0,1]}y^3d_y)={1\over 3}\cdot {1\over 4}={1\over 12}$. The
calculation gives the exact value of the integration, but there is
no evaluation for the function $f(x,y)=x^2y^3$.
Using the computational complexity theory, we study the polynomial
time approximation limitation for the high dimensional integration
for some easily defined functions. In this paper, we consider the
high dimensional integration for multivariate polynomials, which are
defined with format $f(x_1,\cdots, x_d)=p_1(x_1,\cdots,
x_d)p_2(x_1,\cdots, x_d)\cdots p_k(x_1,\cdots, x_d)$, where each
$p_i(x_1,\cdots, x_d)=\sum_{j=1}^d q_j(x_j)$ with polynomial
$q_j(x_j)$ of constant degree. Its integration can be computed in
polynomial space. We show how this problem is related to other hard
problem in the field of computational complexity theory. Therefore,
our model for studying the complexity of high dimensional
integration is totally different from the existing approaches such
as~\cite{SloanWozniakowski97}, and is more general than the old
models. We  show that there is no any factor polynomial time
approximation to the integration problem unless $\Pnew=\NP$.

A similar hardness of approximation result is also derived for the
derivative of the polynomial function. The recent development of
monomial testing
theory~\cite{ChenFu10a,ChenFuLiuSchweller10,ChenFu10c} can be used
to explain the hardness for computing the derivative for a
$\prod\sum\prod$ polynomial.
For the hardness of the approximation for multivariate derivative,
we consider the functions with the format $f(x_1,\cdots,
x_d)=p_1(x_1,\cdots, x_d)p_2(x_1,\cdots, x_d)\cdots p_k(x_1,\cdots,
x_d),$ where each $p_i(x_1,\cdots, x_d)$ is of degree $2$. We also
show that unless $\Pnew=\NP$, there is no any factor polynomial time
approximation to its derivative ${\partial f^{(d)}(x_1,\cdots,
x_d)\over
\partial x_1\cdots
\partial x_d}$ at the origin point $(x_1,\cdots, x_d)=(0,\cdots,0)$. Our results show that the
 high dimension derivative may not be easier than the high dimension
 integration. Since both integration and derivative are widely used,
 this approach may help understand the
 complexity of some mathematics systems that involve high dimension integration or
 derivative.

Partial derivatives were used in developing deterministic algorithms
for the polynomial identity problem (for example,
see~\cite{ShpilkaVolkovich08}), a fundamental problem in the
computational complexity theory. Our intractability result for the
high dimension derivative over multivariate polynomial points out a
barrier of this approach.

Second part of this paper about the inapproximability of derivative
is an application of our recently developed monomial testing
theory~\cite{ChenFu10a,ChenFuLiuSchweller10,ChenFu10c}. It shows
that it is $\#\Pnew$-hard to compute the derivative of a $\prod\sum$
polynomial at the origin point.  We also give some tractable cases
of high dimension integration and derivative.

In section~\ref{overview-sec}, we give an overview about our method
for deriving the inapproximation result of high dimension
integration. The main result of this paper is the inapproximation
for high dimension integration, and is presented in
section~\ref{inapp-integration-sec}. In
section~\ref{inapprox-derivative-sec}, we present the
inapproximation result for high order derivative. Some tractable
cases of high dimension integration and derivative are shown in
section~\ref{tractable-sec}.

\section{Notations}\label{notation-sec}

Let $N=\{0,1,2,\cdots\}$ be the set of all natural numbers. Let
$N^+=\{1,2,\cdots\}$ be the set of all positive natural numbers.

Assume that function $r(n)$ is from $N$ to $N^+$. For a functor
$F(.)$, which converts a multivariate polynomial into a real number,
an algorithm $A(.)$ gives an $r(n)$-factor approximation to $F(f)$
if it satisfies the following conditions:
if $F(f)\ge 0$, then ${F(f)\over r(n)}\le A(f)\le r(n)F(f)$; and
 if
$F(f)<0$, then
 $r(n)F(f)\le A(f)\le {F(f)\over r(n)}$,
 where $n$ is the number of variables in $f$.

Assume that functions $r(n)$ and $s(n)$  are from $N$ to $N^+$. For
a functor $F(.)$, an algorithm $A(.)$ gives an $(r(n),s(n))$-factor
approximation to $F(f)$ such that
if $F(f)\ge 0$, then ${F(f)\over r(n)}-s(n)\le A(f)\le
r(n)F(f)+s(n)$; and
 if
$F(f)< 0$, then
 $r(n)F(f)-s(n)\le A(f)\le {F(f)\over r(n)}+s(n)$,
 where $n$ is the number of variables in $f$.

In this paper, we consider two kinds of functors. The first one is
the integration in the unit cube for a multivariate polynomial:
$\int_{[0,1]^d}f(x_1,\cdots,x_d)d_{x_1}\cdots d_{x_d}$. The second
is the derivative  ${\partial f^{(d)}(x_1,\cdots, x_d)\over \partial
x_1\cdots \partial x_d}$ at the origin point $(x_1,\cdots,
x_d)=(0,\cdots, 0)$.

For the complexity of multivariate integration, we consider the
functions with the format below:
$$f(x_1,\cdots, x_d)=p_1(x_1,\cdots, x_d)p_2(x_1,\cdots, x_d)\cdots
p_k(x_1,\cdots, x_d),$$ where each $p_i(x_1,\cdots,
x_d)=\sum_{j=1}^d q_j(x_j)$ with each single variable polynomials
$q_j(x_j)$ of constant degree. This kind multivariate polynomial is
called $\prod\sum \Si_c$ if the degree of each $q_j(x_j)$ is at most
$c$.

For the complexity of multivariate derivative, we consider the
functions with the format below:
$$f(x_1,\cdots, x_d)=p_1(x_1,\cdots, x_d)p_2(x_1,\cdots, x_d)\cdots
p_k(x_1,\cdots, x_d),$$ where each $p_i(x_1,\cdots, x_d)$ is of a
constant degree. The polynomial $f(x_1,\cdots, x_d)$ is called a
{\it $\prod\sum\prod_k$ polynomial} if the degree of each
$p_i(x_1,\cdots, x_d)$ is at most $k$.

An algorithm is {\it subexponential time} if it runs in
$2^{n^{o(1)}}$ time for all inputs of length $n$. Define $\subE$ to
be the class of languages that have subexponential time algorithms.

\section{Overview of Our Methods}\label{overview-sec}
In this section, we show the brief idea to derive the main result of
this paper (Theorem~\ref{main-theorem}). 3SAT is an NP-complete
problem proved by Cook~\cite{Cook-NP-complete}.
We show that approximating the integration of a $\prod\sum\Si_2$
polynomial is NP-hard by a reduction from 3SAT problem to it. It is
still NP-hard to decide a conjunctive normal form that each variable
appears at most three times with at most one negative time. We
assume that each variable has its negation appears at most one time
(Otherwise, we replace it by its negation).

We show (see Lemma~\ref{orthogonal2-lemma} and equations
(\ref{long-start}) to (\ref{long-end}) at its proof) that there
exist integer coefficients polynomial functions
$g_1(x)=ax^3+bx^2+cx+d$, $g_2(x)=ux+v$, and $f(x)=2x$ satisfy that
$\int_0^1 g_1(x)d_x=1$, $\int_0^1 g_2(x)d_x=1$, $\int_0^1f(x)d_x=1$,
$\int_0^1 g_1(x)g_2(x)d_x=4$,
 $\int_0^1 g_1(x)f(x)d_x=0$, $\int_0^1 g_2(x)f(x)d_x=0$, and $\int_0^1 g_1(x)g_2(x)f(x)d_x=0$.

{\bf Example 1.} Consider the logical formula
$F=(x_1+x_2)(x_1+\overline{x}_2)(\overline{x}_1+x_2)$, which has the
sum of product expansion
$x_1x_1\overline{x}_1+x_1x_1x_2+x_1\overline{x}_2\overline{x}_1+x_1\overline{x}_2x_2+x_2x_1\overline{x}_1+x_2x_1x_2+x_2\overline{x}_2\overline{x}_1+x_2\overline{x}_2x_2$.
The term $x_1x_1x_2$ can bring a truth assignment $x_1=true$ and
$x_2=true$ to make $F$ true. As each variable appears at most $3$
times with at most one negative appearance, the first positive $x_i$
is replaced by $g_1(y_i)$, the second positive $x_i$ is replaced by
$g_2(y_i)$, and the negative $\overline{x}_i$ is replaced by
$f(y_i)$. It is converted into the polynomial
\begin{eqnarray*}
p(y_1,y_2)=(g_1(y_1)+g_1(y_2))(g_2(y_1)+f(y_2))(f(y_1)+g_2(y_2)).
\end{eqnarray*}
The polynomial $p(y_1, y_2)$ has the sum of product expansion
\begin{eqnarray*}
&&g_1(y_1)g_2(y_1)f(y_1)+g_1(y_1)g_2(y_1)g_2(y_2)+g_1(y_1)f(y_2)f(y_1)+g_1(y_1)f(y_2)g_2(y_2)+\\
&&g_1(y_2)g_2(y_1)f(y_1)+g_1(y_2)g_2(y_1)g_2(y_2)+g_1(y_2)f(y_2)f(y_1)+g_1(y_2)f(y_2)g_2(y_2).
\end{eqnarray*}

Consider the integration $\int_{[0,1]^2}p(y_1,y_2)d_{y_1}d_{y_2}$.
The integration can be distributed into those product terms.
$\int_{[0,1]^2}g_1(y_1)g_2(y_1)g_2(y_2)d_{y_1}d_{y_2}$ is one of
them. We have
\begin{eqnarray*}
\int_{[0,1]^2}g_1(y_1)g_2(y_1)g_2(y_2)d_{y_1}d_{y_2}&=&(\int_{[0,1]}g_1(y_1)g_2(y_1)d_{y_1})(\int_{[0,1]}g_2(y_2)d_{y_2})=4\cdot
1=4.
\end{eqnarray*}
 The integrations for other terms are all non-negative integers.
 Thus, $\int_{[0,1]^2}p(y_1,y_2)d_{y_1}d_{y_2}$ is a positive
 integer  due to the
satisfiability of $F$.

{\bf Example 2.} Consider the logical formula
$G=(x_1+x_2)\overline{x}_1\overline{x}_2$, which has the sum of
product expansion
$x_1\overline{x}_1\overline{x}_2+x_1\overline{x}_2\overline{x}_2$.
Neither $x_1\overline{x}_1\overline{x}_2$ nor
$x_1\overline{x}_2\overline{x}_2$ can be satisfied. As each variable
appears at most $3$ times with at most one negative appearance, the
first positive $x_i$ is replaced by $g_1(y_i)$, the second positive
$x_i$ is replaced by $g_2(y_i)$, and the negation case
$\overline{x}_i$ is replaced by $f(y_i)$. It is converted into the
polynomial $q(y_1,y_2)=(g_1(y_1)+g_1(y_2))f(y_1)f(y_2)$. The
polynomial $q(y_1, y_2)$ has the sum of product expansion
$g_1(y_1)f(y_1)f(y_2)+g_1(y_2)f(y_1)f(y_2)$.

Consider the integration $\int_{[0,1]^2}q(y_1,y_2)d_{y_1}d_{y_2}$,
which is identical to
$\int_{[0,1]^2}g_1(y_1)f(y_1)f(y_2)d_{y_1}d_{y_2}+\int_{[0,1]^2}g_1(y_2)f(y_1)f(y_2)d_{y_1}d_{y_2}$.
 We have
\begin{eqnarray*}
\int_{[0,1]^2}g_1(y_1)f(y_1)f(y_2)d_{y_1}d_{y_2}&=&(\int_{[0,1]}g_1(y_1)f(y_1)d_{y_1})(\int_{[0,1]}f(y_2)d_{y_2})=0\cdot
1=0.
\end{eqnarray*}
We also have
 \begin{eqnarray*}
\int_{[0,1]^2}g_1(y_2)f(y_1)f(y_2)d_{y_1}d_{y_2}&=&(\int_{[0,1]}f(y_1)d_{y_1})(\int_{[0,1]}g_1(y_2)f(y_2)d_{y_2})
=1\cdot 0=0.
\end{eqnarray*}

Therefore, $\int_{[0,1]^2}q(y_1,y_2)d_{y_1}d_{y_2}=0$ due to the
unsatisfiability of $G$. Therefore, for any factor $a(n)>0$, a
polynomial time factor $a(n)$-approximation to the integration of a
$\prod\sum\Si_2$ polynomial implies a polynomial time decision for
the satisfiability of the corresponding boolean formula.

\section{Intractability of High Dimensional
Integration}\label{inapp-integration-sec}

In this section, we show that the integration in high dimensional
cube $[0,1]^d$ does not have any factor approximation. We will
reduce an existing NP-complete problem to the integration problem.
Our main technical contribution is in converting a logical formula
into a polynomial. We often use a basic property of integration,
which can be found in some standard text books of calculus (for
example~\cite{Trench78}). Assume function
$f(x_1,\cdots,x_d)=f_1(x_{i_1},\cdots, x_{i_{d_1}})f_2(x_{j_1},
\cdots, x_{j_{d_2}})$, where $\{x_1,\cdots,x_d\}$ is the disjoint
union of $\{x_{i_1},\cdots, x_{i_{d_1}}\}$ and $\{x_{j_1}, \cdots,
x_{j_{d_2}}\}$. Then we have
\begin{eqnarray}
&&\int_{[0,1]^d}f(x_1,\cdots,x_d)d_{x_1}\cdots
d_{x_d}\\
&=&\left(\int_{[0,1]^{d_1}}f_1(x_{i_1},\cdots,x_{i_{d_1}})d_{x_{i_1}}\cdots
d_{x_{i_{d_1}}}\right)\cdot
\left(\int_{[0,1]^{d_2}}f(x_{j_1},\cdots,x_{j_{d_2}})d_{x_{j_1}}\cdots
d_{x_{j_{d_2}}}\right).
\end{eqnarray}

In order to make the conversion from logical operation to algebraic
operation, we represent conjunctive normal form with the following
format. For example, the formula
$(x_1+x_2)(x_1+\overline{x}_2)(\overline{x}_1+x_2)$ is a conjunctive
normal form with two boolean variables $x_1$ and $x_2$, where $+$
represents the logical $\bigvee$, and $.$ represent the logical
$\bigwedge$.

\begin{definition}\scrod
\begin{itemize}
\item
A $3$SAT instance is a conjunctive form $C_1\cdot C_2\cdots C_m$
such each $C_i$ is a disjunction of at most three literals.
\item
$3$SAT is the language of those $3$SAT instances that have
satisfiable assignments.
\item
A $(3,3)$-SAT instance is an instance $G$ for 3SAT such that for
each variable $x$, the total number of times of $x$ and
$\overline{x}$ in $G$ is at most $3$, and the total number of times
of $\overline{x}$ in $G$ is at most $1$.
\item
$(3,3)$-SAT is the language of those $(3,3)$-SAT instances that have
satisfiable assignments.
\end{itemize}
\end{definition}

For examples,
$(x_1+x_2+x_3)(x_1+\overline{x}_2)(\overline{x}_1+x_2)$ is both 3SAT
 and $(3,3)$-SAT instance, and also belongs to both 3SAT and $(3,3)$-SAT. On the
other hand,
$(x_1+x_2+x_3)(\overline{x}_1+\overline{x}_2)(\overline{x}_1+x_2)$
is not a $(3,3)$-SAT instance since $\overline{x}_1$ appears twice
in the formula. The following lemma is similar to a result derived
by Tovey \cite{Tovey84}.


\begin{lemma}\label{3-sat-lemma} There is a polynomial time
reduction $f(.)$ from 3SAT to  $(3,3)$-SAT.
\end{lemma}

\begin{proof}
Let $F$ be an instance for 3SAT.
Let's focus on one variable $x_i$ that appears $m$ times in $F$.
Introduce a series of variables $y_{i,1},\cdots, y_{i,m}$ for $x_i$.
Convert $F$ to $F'$ by changing the $j$-th occurrence of $x_i$ in
$F$ to $y_{i,j}$ for $j=1,\cdots, m$.
Define
\begin{eqnarray*}
G_{x_i}&=&(x_i\rightarrow y_{i,1})\cdot (y_{i,1}\rightarrow y_{i,2})
(y_{i,2}\rightarrow y_{i,3})\cdot (y_{i,3}\rightarrow y_{i,4})
\cdots (y_{i,m-1}\rightarrow y_{i,m})\cdot (y_{i,m}\rightarrow
x_i)\\
&=&(\overline{x}_i+ y_{i,1})\cdot (\overline{y}_{i,1}+y_{i,2})\cdot
(\overline{y}_{i,2}+y_{i,3})\cdot (\overline{y}_{i,3}+y_{i,4})
\cdots (\overline{y}_{i,m-1}+y_{i,m})\cdot(\overline{y}_{i,m}+x_i).
\end{eqnarray*}

Each logical formula $(x\rightarrow y)$ is equivalent to
$(\overline{x}+y)$. If $G_{x_i}$ is true, then  $x_i,
y_{i,1},\cdots, y_{i,m}$ are equivalent.

Convert $F'$ into $F''$ such that $F''=F'G_{x_1}\cdots G_{x_k}$,
where $x_1,\cdots, x_k$ are all variables in $F$.



For each variable $x$ in $F''$ with more than one $\overline{x}$,
create a new variable $y_x$, replace each positive $x$ of $F$ by
$\overline{y_x}$, and each negative $\overline{x}$ by $y_x$. Thus,
$F''$ becomes $F'''$. It is easy to see that $F\in$ 3SAT iff $F''$
is satisfiable iff $F'''\in $(3,3)-SAT.

\end{proof}

\subsection{Integration of $\prod\sum\Si_2$ Polynomial
}

Lemma~\ref{orthogonal2-lemma} is our main technical lemma. It is
used to convert a $(3,3)$-SAT instance into a $\prod\sum\Si_2$
polynomial.

\begin{lemma}\label{orthogonal2-lemma}
There exist integers $b,c,d,u$, and $v$ such that the functions
$g_1(x)=bx^2+cx+d$, $g_2(x)=ux+v$, and $f(x)=2x$ satisfy that
\begin{enumerate}[1.]
\item
$\int_0^1 g_1(x)d_x$, $\int_0^1 g_2(x)d_x$, $\int_0^1f(x)d_x$, and
$\int_0^1 g_1(x)g_2(x)d_x$ are all positive integers, and
\item
  $\int_0^1 g_1(x)f(x)d_x$, $\int_0^1 g_2(x)f(x)d_x$, and $\int_0^1 g_1(x)g_2(x)f(x)d_x$ are all equal to $0$.
\end{enumerate}
\end{lemma}

\begin{proof}
We give the details how to derive the functions $g_1(x)$ and
$g_2(x)$ to satisfy the conditions of the lemma. In order to avoid
solving nonlinear equations, we will fix the two variables $u$ and
$v$ in the early phase of the construction.

\begin{eqnarray}
\int_{[0,1]}f(x)d_x=\int_{[0,1]}2xd_x=x^2|_0^1=1.\label{f-final-eqn}
\end{eqnarray}

\begin{eqnarray}
\int_0^1 g_1(x)d_x&=&\int_0^1 (bx^2+cx+d)d_x\label{g1-start}\\
&=& ({bx^3\over 3}+{cx^2\over 2}+dx)|_0^1\\
&=& {b\over 3}+{c\over 2}+d\\
&=& {1\over 6}(2b+3c+6d).\label{g1-end}
\end{eqnarray}

\begin{eqnarray}
\int_0^1 g_2(x)d_x&=&\int_0^1 (ux+v)d_x\label{g2-start}\\
&=& ({ux^2\over 2}+vx)|_0^1\\
 &=&{1\over 2}(u+2v).\label{g2-end}
\end{eqnarray}

\begin{eqnarray}
\int_0^1 g_2(x)f(x)d_x&=&\int_0^1 (ux+v)2xd_x\label{g2-f-start}\\
&=&2\int_0^1 (ux^2+vx)d_x\\
&=&2({ux^3\over 3}+{vx^2\over 2})|_0^1\\
&=&2({u\over 3}+{v\over 2})\\
&=&{1\over 3}(2u+3v).\label{g2-f-end}
\end{eqnarray}

We let

\begin{eqnarray}
u&=&-6\\
v&=&4.
\end{eqnarray}

Therefore, we have got
\begin{eqnarray}
\int_{[0,1]}g_2(x)d_x&=&1\ \ \ \mbox{(by\ equations\
(\ref{g2-start})\ to\ (\ref{g2-end})),\ and}\label{g2-final-eqn}\\
 \int_{[0,1]}g_2(x)f(x)d_x&=&0\ \ \ \mbox{(by\ equations\
(\ref{g2-f-start})\ to\ (\ref{g2-f-end}))}.\label{g2-f-final-eqn}
\end{eqnarray}

\begin{eqnarray}
\int_0^1 g_1(x)g_2(x)d_x&=&\int_0^1 (bx^2+cx+d)(ux+v)d_x\label{g1-g2-start}\\
&=&\int_0^1  (bx^2+cx+d)(-6x+4)d_x\\
&=&\int_0^1 (-6bx^3-6cx^2-6dx+4bx^2+4cx+4d)d_x\\
&=&\int_0^1 ((-6b)x^3+(4b-6c)x^2+(4c-6d)x+4d)d_x\\
&=&({(-6b)x^4\over 4}+{(4b-6c)x^3\over
3}+{(4c-6d)x^2\over 2}+4dx)|_0^1\\
&=&({(-6b)\over 4}+{(4b-6c)\over 3}+{(4c-6d)\over
2}+4d)\\
&=&{1\over 12}(3\cdot {(-6b)}+4\cdot
{(4b-6c)}+6\cdot {(4c-6d)}+48d)\\
&=&{1\over 12}((-18+16)b+(-24+24)c+(48-36)d)\\
&=&{1\over 12}((-2)b+12d)\\
 &=&{1\over
6}(-b+6d).\label{g1-g2-end}
\end{eqnarray}

\begin{eqnarray}
\int_0^1 g_1(x)f(x)d_x&=&\int_0^1 (bx^2+cx+d)2x d_x\label{g1-f-start}\\
&=&2\int_0^1 (bx^3+cx^2+dx) d_x\\
&=&2({bx^4\over 4}+{cx^3\over 3}+{dx^2\over 2})|_0^1\\
&=&2({b\over 4}+{c\over 3}+{d\over 2})\\
&=&{1\over 6}(3b+4c+6d).\label{g1-f-end}
\end{eqnarray}

\begin{eqnarray}
&&\int_0^1 g_1(x)g_2(x)f(x)d_x\label{g1-g2-f-start}\\
&=&\int_0^1 (bx^2+cx+d)(ux+v)2xd_x\\
&=&2\int_0^1  (bx^2+cx+d)(-6x+4)xd_x\\
&=&2\int_0^1 (-6bx^3-6cx^2-6dx+4bx^2+4cx+4d)xd_x\\
&=&2\int_0^1 ((-6b)x^3+(4b-6c)x^2+(4c-6d)x+4d)xd_x\\
&=&2({(-6b)x^5\over 5}+{(4b-6c)x^4\over
4}+{(4c-6d)x^3\over 3}+{4dx^2\over 2})|_0^1\\
&=&2({(-6b)\over 5}+{(4b-6c)\over 4}+{(4c-6d)\over
3}+{4d\over 2})\\
&=&{2\over 60}({12\cdot (-6b)}+15\cdot {(4b-6c)}+20\cdot
{(4c-6d)}+120d)\\
&=&{1\over 30}((-72+60)b+(-90+80)c+(-120+120)d)\\
&=&{1\over 30}(-12b-10c)\\
&=&{1\over 15}(-6b-5c)\label{g1-g2-f-end}
\end{eqnarray}

We need to satisfy the following conditions:
\begin{eqnarray}
6b+ 5c&=&0 \label{first}\\
3b+4c+6d&=&0\label{second}\\
-b+6d&=&6n_1\ \ \ \mbox{for\ some\ positive\ integer\ } n_1\label{third}\\
2b+3c+6d&=&6n_2\ \ \ \mbox{for\ some\ positive\ integer\ }
n_2\label{fourth}
 \end{eqnarray}
Equation~(\ref{first}) makes $\int_{[0,1]}g_1(x)g_2(x)f(x)d_x=0$
according to equations (\ref{g1-g2-f-start}) to (\ref{g1-g2-f-end}).
Equation (\ref{second}) makes $\int_{[0,1]}g_1(x)f(x)d_x=0$
according to equations (\ref{g1-f-start}) to (\ref{g1-f-end}).
Equation (\ref{third}) makes $\int_{[0,1]}g_1(x)g_2(x)d_x$ be a
positive integer according to equations (\ref{g1-g2-start}) to
(\ref{g1-g2-end}). Equation (\ref{fourth}) makes
$\int_{[0,1]}g_1(x)d_x$ be a positive integer according to equations
(\ref{g1-start}) to (\ref{g1-end}).

Let $x$ and $k$ be integer parameters to be fixed later. We have the
solutions below:
\begin{eqnarray}
b&=&5x\cdot 6^k,\\
c&=&-{6b\over 5}=-6x\cdot 6^k,\ \ \ \ \ \mbox{(by\ equation~(\ref{first}))}\ \ \ \mbox{and}\\
d&=&{1\over 6}(-3b-4c)={1\over 6}(-3\cdot (5x\cdot 6^k)-4(-6x\cdot 6^k))\\
&=&9x\cdot 6^{k-1}\ \ \ \ \ \mbox{(by\ equation~(\ref{second}))}.
\end{eqnarray}

We have the equations:
\begin{eqnarray}
-b+6d&=&-5x\cdot 6^k+6\cdot (9x\cdot 6^{k-1})=4x\cdot 6^k\ \ \ \mbox{and}\label{b-6b-eqn}\\
2b+3c+6d&=&2\cdot (5x\cdot 6^k)+3\cdot (-6x\cdot 6^k)+6\cdot
(9x\cdot 6^{k-1})\\
&=&(10x-18x+9x)6^k=x\cdot 6^k.\label{2b-3c-6d-eqn}
\end{eqnarray}

 \vskip 10pt

Let $x=1$ and $k=1$. We have $b=30, c=-36, $  and $d=9$. We also
have
\begin{eqnarray}
-b+6d&=&24\ \ \  \mbox{(by\ equation\ (\ref{b-6b-eqn}))}\label{g1-g2-final-eqn}\\
2b+3c+6d&=&6\ \ \  \mbox{(by\ equation\
(\ref{2b-3c-6d-eqn}))}.\label{g1-final-eqn}
\end{eqnarray}
 Thus, $g_1(x)=bx^2+cx+d=30x^2-36x+9$,
$g_2(x)=-6x+4$, and $f(x)=2x$. We have the following equations to
satisfy the conditions in the lemma.
\begin{eqnarray}
\int_0^1f(x)d_x&=&1,\ \ \  \mbox{(by\ equation\ (\ref{f-final-eqn})))}\label{long-start}\\
\int_0^1g_1(x)d_x&=&1,\ \ \  \mbox{(by\ equations\ (\ref{g1-start}) to (\ref{g1-end}),\ and\ (\ref{g1-final-eqn}))}\\
\int_0^1g_2(x)d_x&=&1,\ \ \  \mbox{(by\ equation\ (\ref{g2-final-eqn}))}\\
\int_0^1g_1(x)g_2(x)d_x&=&4,\ \ \  \mbox{(by\ equations (\ref{g1-g2-start})\ to\ (\ref{g1-g2-end}),\ and\ (\ref{g1-g2-final-eqn}))}\\
\int_0^1g_1(x)f(x)d_x&=&0,\ \ \  \mbox{(by\ equations\ (\ref{g1-f-start})\ to\ (\ref{g1-f-end}),\ and\ (\ref{second}))}\\
\int_0^1g_2(x)f(x)d_x&=&0,\ \ \  \mbox{(by\
equation(\ref{g2-f-final-eqn}))}, \ \ \ \ \mbox{
and}\\
 \int_0^1g_1(x)g_2(x)f(x)d_x&=&0. \ \ \  \mbox{(by\ equations\ (\ref{g1-g2-f-start}) to
 (\ref{g1-g2-f-end}),\ and\ the\ solutions\ for\ $b$\ and\ $c$))}\label{long-end}
\end{eqnarray}
\end{proof}

\begin{lemma}\label{main2-lemma} There is a  polynomial time
algorithm $h$ such that given a $(3,3)$-SAT instance $s(x_1,\cdots,
x_n)$, it produces a $\prod\sum\Si_2$ polynomial
$h(s(x_1,\cdots,x_n))=p(y_1,\cdots, y_n)$ to satisfy the following
two conditions:
\begin{enumerate}[1.]
\item
if $s(x_1,\cdots, x_n)$ is satisfiable, then
$\int_{[0,1]^{n}}p(y_1,\cdots, y_n)d_{y_1}\cdots d_{y_n}$ is a
positive integer; and
\item
if $s(x_1,\cdots, x_n)$ is not satisfiable, then
$\int_{[0,1]^{n}}p(y_1,\cdots, y_n)d_{y_1}\cdots d_{y_n}$ is zero.
\end{enumerate}
\end{lemma}

\begin{proof} We give two examples to show how a
logical formula is converted into a multivariate polynomial in
section~\ref{overview-sec}. Let polynomials $g_1(y), g_2(y)$, and
$f(y)$ be defined according to those in
Lemma~\ref{orthogonal2-lemma}.

For a $(3,3)$-SAT problem $s(x_1,\cdots, x_n)$, let $p(y_1,\cdots,
y_n)$ be defined a follows.
\begin{itemize}
\item
For the first positive literal $x_i$ in $s(x_1,\cdots, x_n)$,
replace it with $g_1(y_i)$.
\item For the second positive literal $x_i$
in $s(x_1,\cdots, x_n)$, replace it with $g_2(y_i)$.
\item
For the negative literal $\overline{x}_i $ in $s(x_1,\cdots, x_n)$,
replace it with $f(y_i)$.
\end{itemize}

The formula $s(x_1,\cdots, x_n)$ has a sum of product form. It is
satisfiable if and only if one term does not contain a positive and
negative literals for the same variable. If a term contains both
$x_i$ and $\overline x_i$, the corresponding term in the sum of
product for $p(.)$ contains both $g_j(y_i)$ and $f(y_i)$ for some
$j\in\{1,2\}$. This makes it zero after integration by
Lemma~\ref{orthogonal2-lemma}. Therefore, $s(x_1,\cdots, x_n)$ is
satisfiable if and only if $\int_{[0,1]^n}p(y_1,\cdots,
y_n)d_{y_1}\cdots d_{y_n}$ is not zero. Furthermore, it is
satisfiable, the integration is a positive integer by
Lemma~\ref{orthogonal2-lemma}. See the two examples in
section~\ref{overview-sec}. The computational time of $h$ is clearly
polynomial since we convert $s$ to $h(s)$ by replacing each literal
by a single variable function of degree at most $2$.
\end{proof}

\begin{theorem}\label{main-theorem} Let  $a(n)$ be an arbitrary function from $N$ to $N^+$.
Then there is no polynomial time $a(n)$-factor approximation for the
integration of a $\prod\sum\Si_2$ polynomial $p(x_1,\cdots, x_n)$ in
the region $[0,1]^n$ unless $\Pnew=\NP$.
\end{theorem}

\begin{proof} Assume that $A(.)$ is a polynomial time  $a(n)$-factor
approximation for the integration $\int_{[0,1]^n}p(y_1,\cdots,
y_n)d_{y_1}\cdots d_{y_n}$ with $\prod\sum\Si_2$ polynomial
$p(y_1,\cdots, y_n)$. For a $(3,3)$-SAT instance
$s(x_1,\cdots,x_n)$, let $p(y_1,\cdots, y_n)=h(s(x_1,\cdots,x_n))$
according to Lemma~\ref{main2-lemma}. By Lemma~\ref{main2-lemma}, a
$(3,3)$-SAT instance $s(x_1,\cdots,x_n)$ is satisfiable  if and only
if the integration $J=\int_{[0,1]^n}p(y_1,\cdots,y_n)d_{y_1}\cdots
d_{y_n}$ is not zero. Assume that $s(x_1,\cdots,x_n)$ is not
satisfiable, then we have $A(J)\in [J/a(n),J\cdot a(n)]=[0,0]$,
which implies $A(J)=0$. Assume that $s(x_1,\cdots,x_n)$ is
satisfiable, then we have $A(J)\in [J/a(n),J\cdot a(n)]\subseteq
(0,+\infty)$, which implies $A(J)>0$. Thus, $s(x_1,\cdots,x_n)$ is
satisfiable if and only if $A(J)>0$.

Therefore, there is a polynomial time algorithm for solving
$(3,3)$-SAT, which is NP-complete by Lemma~\ref{3-sat-lemma}. So,
$\Pnew=\NP$.
\end{proof}

\begin{theorem} Let  $a(n)$ be an arbitrary function from $N$ to $N^+$.
Then there is no subexponential  time $a(n)$-factor approximation
for the integration of a $\prod\sum\Si_2$ polynomial $p(x_1,\cdots,
x_n)$ in the region $[0,1]^n$ unless $\NP\subseteq\subE$.
\end{theorem}

\begin{proof}Assume that $A(.)$ is a subexponential time  $a(n)$-factor
approximation for the integration $\int_{[0,1]^n}p(y_1,\cdots,
y_n)d_{y_1}\cdots d_{y_n}$ with $\prod\sum\Si_2$ polynomial
$p(y_1,\cdots, y_n)$.

For a $(3,3)$-SAT instance $s(x_1,\cdots,x_n)$, let $p(y_1,\cdots,
y_n)=h(s(x_1,\cdots,x_n))$ according to Lemma~\ref{main2-lemma}. By
Lemma~\ref{main2-lemma}, a $(3,3)$-SAT instance $s(x_1,\cdots,x_n)$
is satisfiable  if and only if the integration
$J=\int_{[0,1]^n}p(y_1,\cdots,y_n)d_{y_1}\cdots d_{y_n}$ is not
zero. Assume that $s(x_1,\cdots,x_n)$ is not satisfiable, then we
have $A(J)\in [J/a(n),J\cdot a(n)]=[0,0]$, which implies $A(J)=0$.
Assume that $s(x_1,\cdots,x_n)$ is satisfiable, then we have
$A(J)\in [J/a(n),J\cdot a(n)]\subseteq (0,+\infty)$, which implies
$A(J)>0$. Thus, $s(x_1,\cdots,x_n)$ is satisfiable if and only if
$A(J)>0$.

Therefore, there is a subexponential time algorithm for solving
$(3,3)$-SAT, which is NP-complete by Lemma~\ref{3-sat-lemma}. Thus,
$\NP\subseteq\subE$.
\end{proof}

\begin{lemma}\label{main3-lemma}
Assume that $a(1^n)$ is a polynomial time computable function from
$N$ to $N^+$ with $a(1^n)>0$ for $n$. There is a polynomial time
algorithm such that given a $(3,3)$-SAT instance $s(x_1,\cdots,
x_n)$, it generates a $\prod\sum\Si_2$ polynomial $p(y_1,\cdots,
y_n)$ such that if $s(x_1,\cdots, x_n)$ is satisfiable, then
$\int_{[0,1]^{n}}p(y_1,\cdots, y_n)d_{y_1\cdots y_n}$ is a positive
integer at least $3a(1^n)^2$; and if $s(x_1,\cdots, x_n)$ is not
satisfiable, $\int_{[0,1]^{n}}p(y_1,\cdots, y_n)d_{y_1\cdots y_n}$
is zero.
\end{lemma}

\begin{proof}
For a (3,3)-SAT problem $s(x_1,\cdots, x_n)$,  let $q(y_1,\cdots,
y_n)=h(s(x_1,\cdots,x_n)$ be constructed as Lemma~\ref{main2-lemma}.

Since $a(1^n)$ is polynomial time computable, let $p(y_1,\cdots,
y_n)=3a(1^n)^2q(y_1,\cdots, y_n)$, which can be computed in a
polynomial time.
\end{proof}

\begin{theorem}
Let $a(1^n)$ be  a polynomial time computable function from $N$ to
$N^+$. Then there is no polynomial time
$(a(1^n),a(1^n))$-approximation for the integration problem
$\int_{[0,1]}f(x_1, \cdots,x_d)d_{x_1}\cdots d_{x_d}$ for  a
$\prod\sum\Si_2$ polynomial $f(.)$ unless $\Pnew=\NP$.
\end{theorem}

 \begin{proof}
Assume that there is a polynomial time
$(a(1^n),a(1^n))$-approximation $App(.)$ for the integration problem
$\int_{[0,1]}f(x_1, \cdots,x_d)d_{x_1}\cdots d_{x_d}$ for  a
$\prod\sum\Si_2$ polynomial $f(.)$.

 Let $s(x_1,\cdots, x_n)$ be
an arbitrary $(3,3)$-SAT instance. Let $p(y_1,\cdots, y_n)$ be the
polynomial according to Lemma~\ref{main3-lemma}.

Let $J=\int_{[0,1]^{n}}p(y_1,\cdots, y_n)d_{y_1\cdots y_n}$. If
$s(x_1,\cdots, x_n)$ is not satisfiable, then $J=0$. Otherwise,
$J\ge 3a(1^n)^2$.

Assume that $s(x_1,\cdots,x_n)$ is not satisfiable.  Since $App(J)$
is an $(a(1^n),a(1^n))$-approximation, we have $App(J)\le {J\cdot
a(1^n)}+a(1^n)=a(1^n)$ by the definition in
section~\ref{notation-sec}.

Assume that $s(x_1,\cdots,x_n)$ is satisfiable.  Since $App(J)$ is
an $(a(1^n),a(1^n))$-approximation, we have $App(J)\ge {J\over
a(1^n)}-a(1^n)\ge {3a(1^n)^2\over a(1^n)}-a(1^n)=2a(1^n)$ by the
definition in section~\ref{notation-sec}.

Therefore, $s(x_1,\cdots,x_n)$ is satisfiable if and only if
$App(J)\ge 2a(1^n)$. Thus, if there is a polynomial time
$(a(1^n),a(1^n))$--approximation, then there is a polynomial time
algorithm for solving $(3,3)$-SAT. By Lemma~\ref{3-sat-lemma},
$\Pnew=\NP$.
 \end{proof}

The well known exponential time hypothesis says $\NP\not\subseteq
\subE$~\cite{ImpagliazzoPaturi99}. Basing on such a hypothesis, we
have the following stronger result about the intractability of high
dimension integration.

\begin{theorem}
Let $a(1^n)$ be  a polynomial time computable function from $N$ to
$N^+$. Then there is no subexponential  time
$(a(1^n),a(1^n))$-approximation for the integration problem
$\int_{[0,1]}f(x_1, \cdots,x_d)d_{x_1}\cdots d_{x_d}$ with  a
$\prod\sum\Si_2$ polynomial $f(.)$ unless $\NP\subseteq \subE$.
\end{theorem}

 \begin{proof}
Assume that there is a subexponential  time
$(a(1^n),a(1^n))$-approximation $App(.)$ for the integration problem
$\int_{[0,1]}f(x_1, \cdots,x_d)d_{x_1}\cdots d_{x_d}$ with  a
$\prod\sum\Si_2$ polynomial $f(.)$.

 Let $s(x_1,\cdots, x_n)$ be
an arbitrary $(3,3)$-SAT instance. Let $p(y_1,\cdots, y_n)$ be the
polynomial according to Lemma~\ref{main3-lemma}. Let
$J=\int_{[0,1]^{n}}p(y_1,\cdots, y_n)d_{y_1\cdots y_n}$. By
Lemma~\ref{main3-lemma}, we have $J\ge 0$.

If $s(x_1,\cdots, x_n)$ is not satisfiable, then $J=0$. Otherwise,
$J\ge 3a(1^n)^2$.

Assume that $s(x_1,\cdots,x_n)$ is not satisfiable.  Since $App(J)$
is an $(a(1^n),a(1^n))$-approximation,  we have $App(J)\le {J\cdot
a(1^n)}+a(1^n)=a(1^n)$ by the definition in
section~\ref{notation-sec}.

Assume that $s(x_1,\cdots,x_n)$ is satisfiable.  Since $App(J)$ is
an $(a(1^n),a(1^n))$-approximation, we have $App(J)\ge {J\over
a(1^n)}-a(1^n)\ge {3a(1^n)^2\over a(1^n)}-a(1^n)=2a(1^n)$ by the
definition in section~\ref{notation-sec}.

Therefore, $s(x_1,\cdots,x_n)$ is satisfiable if and only if
$App(J)\ge 2a(1^n)$. Thus, if there is a subexponential  time
$(a(1^n),a(1^n))$-approximation, then there is a subexponential time
algorithm for solving $(3,3)$-SAT. By Lemma~\ref{3-sat-lemma},
$\NP\subseteq \subE$.
 \end{proof}

\section{Inapproximation of
Derivative}\label{inapprox-derivative-sec}

In this section, we study the hardness of high dimensional
derivative. We derive the inapproximation results under both
$\NP\not=\Pnew$ and $\NP\not\subseteq \subE$ assumptions.

\begin{definition}
A monomial is an expression $x_1^{a_1}\cdots x_d^{a_d}$ and its
degree is $a_1+\cdots+a_d$. A monomial $x_1^{a_1}\cdots x_d^{a_d}$,
in which $x_1,\cdots, x_d$ are different variables, is a {\it
multilinear} if $a_1=a_2=\cdots=a_d=1$.
\end{definition}

For example, $(x_1x_3+x_2^2)(x_2x_4+x_3^2)$ is a $\prod\sum\prod_2$
polynomial. It has a multilinear monomial $x_1x_2x_3x_4$ in its sum
of products expansion.

We give Lemma~\ref{derivative-lemma} to convert an instance $f$ for
$(3,3)$-SAT into a $\prod\sum\prod_2$ polynomial. The technology
developed in~\cite{ChenFu10a,ChenFu10c} will be applied in the
construction.

\begin{lemma}\label{derivative-lemma}
Let $a(1^n)$ be a polynomial time computable function from $N$ to
$N^+$.
 Then there is a polynomial time algorithm $A$ such that given a
$(3,3)$-SAT instance $F(y_1,\cdots, y_d)$, the algorithm returns a
$\prod\sum\prod_2$ polynomial $G(x_1,\cdots, x_n)$ such that
\begin{enumerate}[1.]
\item
If $F$ is not satisfiable, then $G$ does not have a multiliear
monomial with an nonzero coefficient in its sum of product
expansion.
\item
If $F$ is satisfiable, then $G$ has the multiliear monomial
$x_1\cdots x_n$ with a positive integer coefficient at least
$3a(1^n)^2$ in its sum of product expansion.
\end{enumerate}
\end{lemma}

\begin{proof}
Let $(3,3)$-SAT instance $F$ be $C_1C_2\cdots C_k$. Each clause
$C_i$ has format $y_{i_1}^*+y_{i_2}^*+ y_{i_3}^*$, where literal
$y_j^*$ is either $y_j$ or its negation $\overline{y}_j$. Since $F$
is a $(3,3)$-SAT instance, for each variable $y_i$ in $F$, $y_i$ and
$\overline{y}_i$ totally appear at most three times in $F$, and
$\overline{y}_i$ appears at most once in $F$.

For each variable $y_i$ in $F$, create four new variables $z_{i,1}$,
$z_{i,2}$, $u_{i,1}$ and $u_{i,2}$.  Convert formula $F$ into
polynomial $G_1$ such that for each $y_i$ in $F$, the first positive
occurrence $y_i$  is changed into $z_{i,1}u_{i,1}$, the second
positive occurrence $y_i$ is changed into $z_{i,2}u_{i,2}$, and the
negative occurrence $\overline{y}_i$ is changed to $z_{i,1}z_{i,2}$.
After the conversion for all the variables, formula $F$ is
transformed into a polynomial $G_1 $. We have that $F$ is
satisfiable if and only if $G_1 $ has a multilinear monomial with
positive coefficient in its sum of products expansion. This is
because a multiliear monomial in the sum of product expansion of
$G_1$ corresponds a consistent conjunctive term, which does not
contain both $y_i$ and its negation $\overline{y}_i$ for some
variable $y_i$,  in the sum of product expansion of $F$.


Let $H_1$ be the set of all variables in $G_1 $. Assume that $d_1$
is the degree of $G_1 $ (it is easy to see that all monomials in the
sum-product expansion of $G_1$ have the same degree $d_1$). Let $m$
be the number of variables in $H_1$. Let $d$ be the number of
boolean variables in $F$. Assume that no clause $C_i$ in $F$
contains a single literal (otherwise, we can force the literal to be
true to simplify $F$). The number of clauses in $F$ is at most
${3d\over 2}$ since each variable appears in $F$ at most three times
and each clause $C_i$ of $F$ contains at least two literals. The
degree $G_1$ is at most $3d$ since each literal of $F$ is replaced
by a product of two variables. The number of variables in $G_1$ is
$m=4d$ (we create four new variables for each variable in $F$),
which is larger than the degree of $G_1$.

Create new variables $v_1,\cdots, v_{m-d_1}$.
 For $j=1,\cdots, m-d_1$, let $q_j=\sum_{x\in H_1} xv_j$.
 Finally, we get the polynomial $G=3a(1^n)^2\cdot  G_1\cdot q_1\cdots q_{m-d_1}$, where $n=m+(m-d_1)$.  Note that $3a(1^n)^2$ in
the polynomial $G$ is considered an integer constant which does not
contain any variable. The degree of $G$ is
 $n=d_1+2(m-d_1)=m+(m-d_1)$.
Thus, the
 degree of $G$ is the same as the total number of variables in $g$.
 We can show that
$f$ is satisfiable if and only if there is a multilinear monomial,
which is the product of all variables in $G$, with positive
coefficient of size at least $3a(1^n)^2$.
\end{proof}

\begin{theorem}Assume that $r(n)$ is a function from $N$ to $N^+$.
If there is a polynomial time algorithm $A$ such that given a
$\prod\sum\prod_2$ polynomial $g(x_1,\cdots,x_n)$,  it gives an
$r(n)$-factor approximation to ${\partial
g^{(n)}(x_1,\cdots,x_n)\over
\partial x_1\cdots \partial x_n}$ at the origin point $(x_1,\cdots,x_n)=(0,\cdots,0)$, then $\Pnew=\NP$.
\end{theorem}

\begin{proof} Assume that $A(.)$ is a polynomial time
$r(n)$-approximation for computing ${\partial
g^{(n)}(x_1,\cdots,x_n)\over
\partial x_1\cdots \partial x_n}$ at the origin point
$(x_1,\cdots,x_n)=(0,\cdots,0)$.

 Assume that $f$ is an arbitrary
formula in a $(3,3)$-SAT problem. By Lemma~\ref{derivative-lemma},
we can get a polynomial $g(x_1,\cdots,x_n)$.
 The derivative ${\partial g^{(n)}(0,\cdots,0)\over
\partial x_1\cdots \partial x_n}$ is equal to the coefficient of $x_1\cdots
x_n$ in the sum of product expansion of $g$.

If $f$  is satisfiable, we have $A(g)>0$, and if $f$  is not
satisfiable, we have $A(g)=0$ since $A(.)$ is a $r(n)$-approximation
and $r(n)\ge 1$. We can know if the coefficient of $x_1\cdots x_n$
in the sum of product expansion of $g$ is positive in polynomial
time. Thus, $(3,3)$-SAT is solvable in polynomial time. Since
$(3,3)$-SAT is NP-complete, we have $\Pnew=\NP$.
\end{proof}

 Basing on exponential time hypothesis  $\NP\not\subseteq \subE$~\cite{ImpagliazzoPaturi99}, we have the following
stronger result about the intractability of high dimension
derivative.

\begin{theorem}Assume that $r(n)$ is a function from $N$ to $N^+$.
If there is a subexponential  time algorithm $A$ such that given a
$\prod\sum\prod_2$ polynomial $g(x_1,\cdots,x_n)$, it gives an
$r(n)$-factor approximation to ${\partial
g^{(n)}(x_1,\cdots,x_n)\over
\partial x_1\cdots \partial x_n}$ at the origin point $(x_1,\cdots,x_n)=(0,\cdots,0)$, then $\NP\subseteq \subE$.
\end{theorem}

\begin{proof}
 Assume that $A(.)$ is a subexponential time
$r(n)$-approximation for computing ${\partial
g^{(n)}(x_1,\cdots,x_n)\over
\partial x_1\cdots \partial x_n}$ at the origin point
$(x_1,\cdots,x_n)=(0,\cdots,0)$.

 Assume that $f$ is an arbitrary
formula in a $(3,3)$-SAT problem. By Lemma~\ref{derivative-lemma},
we can get a polynomial $g(x_1,\cdots,x_n)$.
 The derivative ${\partial g^{(n)}(0,\cdots,0)\over
\partial x_1\cdots \partial x_n}$ is equal to the coefficient of $x_1\cdots
x_n$ in the sum of product expansion of $g$.

If $f$  is satisfiable, we have $A(g)>0$, and if $f$  is not
satisfiable, we have $A(g)=0$ since $A(.)$ is a $r(n)$-approximation
and $r(n)\ge 1$. We can know if the coefficient of $x_1\cdots x_n$
in the sum of product expansion of $g$ is positive in subexponential
time. Thus, $(3,3)$-SAT is solvable  in subexponential time. Since
$(3,3)$-SAT is NP-complete, we have $\NP\subseteq \subE$.
\end{proof}

\begin{theorem}\label{dula-approx-derivative-theorem}
Let $a(1^n)$ be a polynomial time computable function from $N$ to
$N^+$. Then there is no polynomial time
$(a(1^n),a(1^n))$-approximation for ${\partial
g^{(n)}(x_1,\cdots,x_n)\over
\partial x_1\cdots \partial x_n}$ with $g(x_1,\cdots,x_n)$ as
 a $\prod\sum\prod_2$ polynomial at the origin point $(x_1,\cdots,x_n)=(0,\cdots,0)$, unless  $\Pnew=\NP$.
\end{theorem}

\begin{proof}
Assume that $App(.)$ is a polynomial time
$(a(1^n),a(1^n))$-approximation for computing ${\partial
g^{(n)}(x_1,\cdots,x_n)\over
\partial x_1\cdots \partial x_n}$ at the origin point
$(x_1,\cdots,x_n)=(0,\cdots,0)$.

Given an arbitrary $(3,3)$-SAT instance $F$, let the
$\prod\sum\prod_2$ polynomial $G(x_1,\cdots, x_n)$
 be constructed according to
Lemma~\ref{derivative-lemma}.

 If $F$ is not satisfiable, then we have
${\partial G^{(n)}(x_1,\cdots,x_n)\over
\partial x_1\cdots \partial x_n}$ is zero at the origin point
$(x_1,\cdots,x_n)=(0,\cdots,0)$. Otherwise, it is at least
$3a(1^n)^2$.

Assume that $F$ is not satisfiable.  Since $App(J)$ is an
$(a(1^n),a(1^n))$-approximation, we have $App(J)\le {J\cdot
a(1^n)}+a(1^n)=a(1^n)$ by the definition in
section~\ref{notation-sec}.

Assume that $F$ is satisfiable.  Since $App(J)$ is an
$(a(1^n),a(1^n))$-approximation, we have $App(J)\ge {J\over
a(1^n)}-a(1^n)\ge {3a(1^n)^2\over a(1^n)}-a(1^n)=2a(1^n)$ by the
definition in section~\ref{notation-sec}.

Therefore, $F$ is satisfiable if and only if $App(J)\ge 2a(1^n)$.
Thus, if there is a polynomial time $(a(1^n),a(1^n))$-approximation,
then there is a polynomial time algorithm for solving $(3,3)$-SAT.
By Lemma~\ref{3-sat-lemma}, $\Pnew=\NP$.
 \end{proof}

\begin{theorem}
Let $a(1^n)$ be a polynomial time computable function from $N$ to
$N^+$. Then there is no subexponential  time
$(a(1^n),a(1^n))$-approximation for ${\partial
g^{(n)}(x_1,\cdots,x_n)\over
\partial x_1\cdots \partial x_n}$ at the origin point $(x_1,\cdots,x_n)=(0,\cdots,0)$, unless  $\NP\subseteq \subE$.
\end{theorem}

\begin{proof}
The proof is similar to that of
Theorem~\ref{dula-approx-derivative-theorem}.
\end{proof}



\section{Some Tractable Integrations and
Derivatives}\label{tractable-sec}
 In this section, we present some
polynomial time algorithms for integration with some restrictions.
We also show a case that the derivative can fully polynomial time
approximation scheme.

\subsection{Bounded Width Product}
\begin{definition}
A formula $f_1\cdot f_2\cdots f_m$ is $c$-wide if for each variables
$x_i$, there is an index $j$ such that $x_i$ only appears in $f_j,
f_{j+1},\cdots f_{j+c-1}$, where each $f_i$ is a sum of monomials.
\end{definition}

\begin{theorem}
There is an $O(mn^{3c})$ time algorithm to compute the integration
$\int_{[0,1]^d} F(x_1,\cdots, x_d)$ for a $c$-wide formula
$F(x_1,\cdots, x_d)=f_1\cdots f_m$, where $n$ is the total length of
$F$.
\end{theorem}

\begin{proof}
Apply the divide and conquer method.
Convert $F$ into $F_1 G F_2$ such that $G$ is a product of at most
$c$ sub-formulas $f_i\cdots f_j$ with $j-i=c$ in the middle region
of $F$ (we can let $i=\ceiling{m-c\over 2}+1,$ and
$j=\ceiling{m-c\over 2}+c$).

Let $S_1$ be the set of variables that are only in $F_1$, $S_2$ be
the set of variables that are only in $F_2$,  and $S$ be the set of
variables that appear in $G$. The set of variables in $F$ is
partitioned into $S_1, S$, and $S_2$.

As $F_1=f_1\cdots f_{i-1}$, we convert $F_1$ into $F_1^*=f_1\cdots
f_{i-c}f_1^*$, where $f_1^*$ is the product of the last $c$
sub-formulas: $f_1^*=f_{i-c}\cdots f_{i-1}$. Similarly, as
$F_2=f_{j+1}\cdots f_{m}$, we convert $F_2$ into
$F_2^*=f_2^*f_{j+c}\cdots f_m$, where $f_2^*$ is the product of the
first $c$ sub-formulas: $f_2^*=f_{j+1}\cdots f_{j+c}$. Convert $G$
into the sum of products.

We have
\begin{eqnarray*}
\int_{[0,1]^d} F(x_1,\cdots, x_d)d_{x_1}\cdots d_{x_d}&=&
\int_{[0,1]^{|S|}} G\cdot (\int_{[0,1]^{|S_1|}} F_1d_{S_1}) \cdot
(\int_{[0,1]^{|S_2|}}
F_2d_{S_2})d_S\\
&=&\int_{[0,1]^{|S|}} G\cdot (\int_{[0,1]^{|S_1|}} F_1^*d_{S_1})
\cdot (\int_{[0,1]^{|S_2|}} F_2^*d_{S_2})d_S.
\end{eqnarray*}
 The integration
$\int_{[0,1]^{|S_1|}} F_1d_{S_1}$ can be expressed as a polynomial
of variables in $S$. The integration $\int_{[0,1]^{|S_2|}}
F_2d_{S_2}$ can be expressed as a polynomial of variables in $S$.

We have the recursive equation for the computational time
$T(m)=2T(m/2)+ O(n^{3c})$. This gives $T(m)=O(mn^{3c})$.
\end{proof}

\subsection{Tractable Derivative}

In this section, we show that computing the derivative of a class of
$\prod\sum$ polynomial is $\#P$-hard, and also give a polynomial
time randomized approximation scheme by using the theory of testing
monomials developed by Chen and Fu~\cite{ChenFu10a,ChenFu10c}.

\begin{definition}
Let $f(x_1,\cdots, x_d)=p_1(x_1,\cdots,x_d)\cdots
p_k(x_1,\cdots,x_d)$ be a $\prod\sum$ polynomial. If for each
$p_i(x_1,\cdots,x_d)$, each variable's coefficient is either $0$ or
$1$, then $f$ is called a $\prod\sum^*$ polynomial.
\end{definition}

We show that the derivative for a $\prod\sum^*$ polynomial has a
polynomial time approximation scheme.  Chen and Fu derived the
following theorem by a reduction from the number of perfect
matchings in a bipartite.

\begin{theorem}[Chen and Fu~\cite{ChenFu10c}]\label{FuChen-theorem}\scrod
\begin{enumerate}[1.]
\item
There is a polynomial time randomized algorithm to approximate the
coefficient of a $\prod\sum^*$ polynomial.
\item
It is $\#P$-hard to compute the coefficient of the multilinear
$x_1\cdots x_d$ in a $\prod\sum^*$ polynomial $f(x_1,\cdots, x_d)$.
\end{enumerate}
\end{theorem}

Theorem~\ref{FuChen-theorem} implies
Theorem~\ref{derivative-app-theorem}.

\begin{theorem}\label{derivative-app-theorem}\scrod
\begin{enumerate}[1.]
\item
Let $\epsilon$ be an arbitrary constant in $(0,1)$. Then there is a
polynomial time randomized algorithm that given a $\prod\sum^*$
polynomial $f$, it returns a $(1+\epsilon)$-approximation for
${\partial f(x_1,\cdots, x_d)^{(d)}\over \partial x_1\cdots\partial
x_d}$ at the point $(0,\cdots,0)$.
\item
It is $\#P$-hard to compute ${\partial f(x_1,\cdots, x_d)^{(d)}\over
\partial x_1\cdots\partial x_d}$ at the point $(0,\cdots,0)$ for a $\prod\sum^*$
polynomial $f$.
\end{enumerate}
\end{theorem}

\begin{proof}
For a $\prod\sum^*$ polynomial $f(x_1,\cdots,x_d)$, its ${\partial
f(x_1,\cdots, x_d)^{(d)}\over
\partial x_1\cdots\partial x_d}$ at the point $(0,\cdots,0)$ is
identical to the coefficient of the monomial $x_1\cdots x_d$ in the
sum of products in the expansion of $f(x_1,\cdots,x_d)$. The theorem
follows from Theorem~\ref{FuChen-theorem}.
\end{proof}


\section{Conclusions}

Using the theory of NP-hardness, we characterize the intractability
of approximation for two fundamental mathematical operations:
Integration and derivative in high dimensional space. We may see
that this approach will be applied to determining the computational
complexity of more mathematics systems that involve integration and
derivative. We show that derivative for $\prod\sum\prod_2$ is
$\#P$-hard. Both integration and derivative for $\prod\sum\prod$
polynomials are in the class $\#P$. This shows that derivative is
not easier than integration in the high dimension.


\end{document}